\documentclass[Danish,UKEnglish]{llncs}
\usepackage{todonotes}
\usepackage{stmaryrd}
\usepackage{amssymb}
\usepackage{mathtools}
\usepackage{paralist}

\usepackage{thm-restate}

\usepackage[caption=false]{subfig}

\usepackage{tikz}
\usetikzlibrary{positioning,shapes,fit,arrows}
\usetikzlibrary{shapes.geometric}
\newcommand{\NF}[1]{#1_{\mathit{NF}}}

\usepackage{soul}

\def\rightarrowstar{\rightarrow^*}
\title{Confluence and Convergence in Probabilistically Terminating Reduction Systems}
\titlerunning{Confluence in Probabilistically Terminating Reduction Systems} 

\author{Maja H. Kirkeby \and Henning Christiansen}
\institute{Computer Science, Roskilde University, Denmark\\  \email{majaht@ruc.dk} and \email{henning@ruc.dk}}

\begin{document}
\maketitle

\begin{abstract}
Convergence of an abstract reduction system (ARS) is the property that any derivation from an initial state will end in the same final state, a.k.a.\ normal form.
We generalize this for probabilistic ARS as almost-sure convergence,
meaning that the normal form is reached with probability one, even
if diverging derivations may exist.
We show and exemplify properties that can be used for proving almost-sure convergence of probabilistic ARS, generalizing known results from ARS.

%
%
\end{abstract}

\section{Introduction}
Probabilistic abstract reduction systems, PARS, are general
models of systems that develop over time in discrete steps~\cite{DBLP:conf/rta/BournezK02}.
In each non-final state, the choice of successor state is governed by a
probability distribution, which in turn induces a global, probabilistic
behaviour of the system.  
Probabilities make termination more than a simple yes-no question,
and the following criteria have been proposed:
 \emph{probabilistic termination} --
a derivation terminates with some probability $> 0$ -- and
\emph{almost-sure termination} -- a derivation terminates with probability $= 1$,
even if infinite derivations may exist (and whose total probability thus amounts to 0).
When considering a PARS as a computational system, almost-sure termination
may be the most interesting, and there exist well-established methods for
proving this property~\cite{Bournez05,Fioriti2015}.

PARS cover a variety of probabilistic algorithms and
programs, scheduling strategies and protocols~\cite{Baier2008,DBLP:conf/rta/BournezK02,Rabin1982}, and PARS is
a well-suited abstraction level for better understanding their termination
and correctness properties.
Randomized or probabilistic algorithms (e.g.,~\cite{Babai1979,Maffioli1984,MotwaniRaghavan1995}) come in two groups:  Monte Carlo Algorithms that allow a set of alternative outputs (typically only correct with a certain probability or within a certain accuracy), e.g., Karger-Stein's Minimum Cut~\cite{DBLP:journals/jacm/KargerS96}, Monte Carlo integration and Simulated Annealing~\cite{DBLP:journals/science/KirkpatrickGV83}; and  Las Vegas Algorithms, that provide one (correct) output and that may be simpler and on average more efficient than their deterministic counterparts, e.g., Randomized Quicksort~\cite{DBLP:journals/jacm/FrazerM70}, checking equivalence of circular lists~\cite{DBLP:journals/ipl/Itai79}, probabilistic modular GCD~\cite{DBLP:conf/eurosam/Zippel79}.
We focus on results that are relevant for the latter kind of
systems, and here the property of \emph{convergence} is interesting,
as it
may be a necessary condition for correctness: a system is convergent if it is guaranteed to terminate {\small\hbox to 0pt{\vbox to 0pt{\vbox to 50pt{\vfill\hrule\vskip 4pt
\hbox{\small  {\it Preliminary proceedings of LOPSTR 2017.} October 10--12, 2017, Namur, Belgium}
\hbox{\small The research leading to this paper
      is supported by The Danish Council for}
\hbox{\small  Independent Research, Natural Sciences,
      grant no.~DFF 4181-00442.}
}\vss}}}with a unique result.
We introduce  the notion of
\emph{almost-sure convergence} for PARS,
meaning that a unique result is found with probability $=1$,
although there may be diverging computations;
this property is a necessary condition for partial correctness,
more precisely a strengthened version of partial correctness
where the probability of not getting a result is zero.

The related notion of {\em confluence} has been extensively studied
for 
ARS, e.g.,~\cite{BaderNipkow1999,Huet1980}
, and especially for terminating ones
for which confluence implies convergence: a system is confluent if, whenever alternative paths (i.e., repeated reductions; computations) are possible from some state, these paths can be extended to join in a common state.
Newman's lemma~\cite{Newman42} from 1942 is one of the most central results: in a terminating system,
confluence (and thus convergence) can be shown from a simpler property called local confluence.
In, e.g., term rewriting~\cite{BaderNipkow1999} and (a subset of) the programming language CHR~\cite{Abdennadher1997,DBLP:conf/cp/AbdennadherFM96},
proving local confluence
may be reduced to a finite number of cases, described by \emph{critical pairs} (for a definition, see these references), which in some cases may be checked automatically.
It is well-known that Newman's lemma does not generalize to non-terminating systems (and thus neither to almost-sure terminating ones); see, e.g.,~\cite{Huet1980}.

Probabilistic and almost-sure versions of confluence
were introduced 
concurrently by Fr\"uhwirth et~al.~\cite{Fruhwirth2002} -- in the context of
a probabilistic version of  CHR -- and by
Bournez and Kirchner~\cite{DBLP:conf/rta/BournezK02} in more generality for PARS.
However, the definitions in the latter reference were
given indirectly, assuming a deep insight into Homogeneous Markov Chain Theory, and a number of central properties were listed without hints of proofs.

In the present paper, we consider the important property of almost-sure convergence for
PARS and state
properties that are relevant for proving it.
In contrast to~\cite{DBLP:conf/rta/BournezK02}, our definitions
are self-contained, based on elementary math, and proofs
are included.
One of our main and novel results is that almost-sure termination together with confluence (in the classical sense) 
gives almost-sure convergence.
Almost-sure convergence and almost-sure termination
were introduced  in an early 1983
paper~\cite{Hart1983} for a specific class of probabilistic programs with finite
state space, but our generalization to PARS'  appears to be new.

In 1991, Curien and Ghelli~\cite{CurienGhelli1991} described a powerful method for proving confluence of non-probabilistic systems,
using suitable transformations from the original system into one, known to be confluent.
We can show how this result applies to probabilistic systems, and we develop an analogous method for also proving non-confluence.

In section~\ref{sec:basics}, we review definitions for abstract reduction systems
and introduce and motivate our choices of definitions for their probabilistic counterparts;
a proof that the defined probabilities actually constitute a probability distribution is found in the Appendix.
Section~\ref{sec:props-of-PARS} formulates and proves important properties, relevant for showing {almost-sure convergence} of particular systems. Section~\ref{sec:transformation} goes in detail with applications of the transformational approach~\cite{CurienGhelli1991} to (dis-) proving {almost-sure convergence}, and in Section~\ref{sec:examples} we demonstrate the use of this for a random walk system and Hermans' Ring.
We add a few more comments on selected, related work in section~\ref{sec:related},
and section~\ref{sec:conclusion} provides a summary and suggestions for future work.

\section{Basic definitions}\label{sec:basics}
The definitions for non-probabilistic systems are standard; see, e.g.,~\cite{Huet1980,BaderNipkow1999}.

\begin{definition}[ARS]
An \emph{Abstract Reduction System} is a pair $R = (A, \rightarrow)$ where the \emph{reduction} $\rightarrow$ is a 
binary relation on a countable set $A$.
\end{definition}\noindent
Instead of $(s,t) \in\, \rightarrow $ we write $s \rightarrow t$ (or $t \leftarrow s$ when convenient),
and $s \rightarrowstar t$ denotes the transitive reflexive closure of $\rightarrow$.

In the literature, an ARS is often required to have only finite branching. i.e., for any element $s$, the set $\{t\mid s\rightarrow t\}$
is finite. We do not require this, as the implicit restriction to countable branching is sufficient for our purposes.


The set of \emph{normal form}s $\NF{R}$ are those $s\in A$ for which there is no $t \in A$ such that $s \rightarrow t$.
For given element $s$, the \emph{normal forms of $s$}, are defined as the set $\NF{R}(s)= \{t\in \NF{R}\mid s\rightarrowstar t\}$.
An element which is not a normal form is said to be 
\emph{reducible}; i.e., an element $s$ is reducible if and only if  $\{ s' \mid s \rightarrow s'\} \neq \emptyset$.

A \emph{path} from an element $s$ is a (finite or infinite) sequence of reductions $s \rightarrow s_1 \rightarrow s_2 \rightarrow \cdots$; a finite path $s \rightarrow s_1 \rightarrow s_2 \rightarrow \cdots  \rightarrow s_n$ has \emph{length $n$} ($n \geq 0$);
in particular, we recognize an empty path (of length $0$) from a given state to itself.
For given elements $s$ and $t\in\NF{R}(s)$, $\Delta(s,t)$ denotes the set of finite paths $s \rightarrow \cdots \rightarrow t$ (including the empty path);
$\Delta^{\infty}(s)$ denotes the set of infinite paths from $s$. A system is
\begin{itemize}
\item \emph{confluent} if for all $s_1 \leftarrow^{*} s \rightarrow^{*} s_2$ there is a $t$ such that $s_1 \rightarrow^{*} t \leftarrow^{*} s_2$,
\item \emph{locally confluent} if for all $s_1 \leftarrow s \rightarrow s_2$ there is a $t$ such that $s_1 \rightarrow^{*} t \leftarrow^{*} s_2$,
\item \emph{terminating}\footnote{A terminating system is also called \emph{strongly normalizing} elsewhere, e.g.,~\cite{CurienGhelli1991}.} iff it has no infinite path, 
\item \emph{convergent} iff it is terminating and confluent, 
and 
\item \emph{normalizing}\footnote{A normalizing system is also called \emph{weakly normalizing} or \emph{weakly terminating} elsewhere, e.g.,~\cite{CurienGhelli1991}. } iff every element $s$ has a normal form, i.e., there is an element $t \in \NF{R}$ such that $s \rightarrowstar t$. 
\end{itemize}
\par\noindent
Notice that a normalizing system may not be terminating. A fundamental result for ARS is Newman's Lemma: a terminating system is confluent if and only if it is locally confluent. 

\medskip\noindent
The following property indicates the complexity of the probability measures that are needed in order to cope
with paths in probabilistic abstract reduction systems defined over countable sets.
\begin{proposition}
Given an ARS as above and given elements $s$ and $t\in\NF{R}(s)$, it holds that $\Delta(s,t)$ is countable,
and $\Delta^{\infty}(s)$ may or may not be countable.
\end{proposition}

\begin{proof}
For the first part, $\Delta(s,t)$ is isomorphic to a subset of
$\bigcup_{n=1,2,\ldots} A^n$. A countable union of countable sets is countable, so $\Delta(s,t)$ is countable.

For the second part,  consider the ARS $\langle \{0,1\}, \{i\rightarrow j \mid i,j\in \{0,1\}\}\rangle$.
Each infinite path can be read as a real number in the unit interval, and any such real number can be
described by an infinite path. The real numbers are not countable. 
\end{proof}
This means that we can define discrete and summable probabilities over $\Delta(s,t)$, and -- which we will avoid --
considering probabilities over the space $\Delta^{\infty}(s)$ requires a more advanced measure.

In the next definition, a path is considered a Markov process/chain, i.e., each reduction step is independent of the previous ones, and thus the probability of a path is defined as a product in the usual way. PARS can be seen as a special
case of Homogenous Markov Chains, cf.~\cite{DBLP:conf/rta/BournezK02}, but for practical reasons
it is relevant to introduce them as generalizations of ARS.

\begin{definition}[PARS]\label{def:pars}
 A \emph{Probabilistic Abstract Reduction System} is a pair  $R^P = (R, P)$  where $R = (A, \rightarrow)$ is an ARS, and for each reducible element $s \in A\setminus\NF{{R}}$, $P(s \rightarrow \cdot)$ is a probability distribution over the reductions from $s$, i.e.,
  $\sum_{s \rightarrow t} P(s \rightarrow t) = 1$;
it is assumed, that for all $s$ and $t$, $P(s\rightarrow t)>0$ if and only if $s\rightarrow t$.

The \emph{probability of a finite path} $s_0 \rightarrow{} s_1 \rightarrow{}{} \ldots \rightarrow{}{} s_n$ with $n\geq 0$ is given as %
$${P(s_0 \rightarrow s_1  \rightarrow \ldots  \rightarrow s_n) = \prod_{i=1}^n  P(s_{i-1}  \rightarrow s_i) }.$$
For any element $s$ and normal form $t\in\NF{R}(s)$, the \emph{probability of $s$ reaching  $t$},
written $P(s \rightarrow^* t)$, is defined as
$$P(s \rightarrow^* t) = \sum_{\delta \in \Delta(s,t)} P(\delta);$$
the \emph{probability of $s$ not reaching a normal form} (or \emph{diverging}) is defined as
$$P(s \rightarrow^\infty) = 1 - \sum_{t \in \NF{R}(s)} P(s \rightarrow^* t).$$
When referring to \emph{confluence}, \emph{local confluence}, \emph{termination}, and \emph{normalization} of a PARS, we refer to these properties for the underlying ARS.
\end{definition}
Notice that when $s$ is a normal form then $P(s \rightarrow^* s)=1$ since $\Delta(s,t)$ contains only the empty path with probability $\prod_{i=1}^0  P(s_{i-1}  \rightarrow s_i)=1$.
It is important that $P(s \rightarrow^* t)$ is defined only when $t$ is a normal form of $s$ since otherwise,
the defining sum may be $\ge 1$, as demonstrated by the following example.
\begin{example}\label{ex:notP}
Consider the PARS $R^P$ given in Figure \ref{fig:examples-abc}\protect\subref{fig:RP}; formally, 
$R^P = ((\{0,1\},$\break$\{0\!\shortrightarrow \!1, 1\!\shortrightarrow \!1\}),\;P)$ with $P(0\!\shortrightarrow\!1) = 1$ and $P(1\!\shortrightarrow\!1) = 1$.
An attempt to define $P(0 \!\rightarrow^*\! 1)$ as in Def.~\ref{def:pars},
for the reducible element $1$, does not lead to a probability, i.e., $P(0 \!\rightarrow^*\! 1) \not\leq 1$:
$P(0 \!\rightarrow^*\! 1) = P(0\!\shortrightarrow\!\! 1)+ P(0\!\shortrightarrow\!\! 1\!\!\shortrightarrow\!\! 1)+P(0\!\shortrightarrow\! 1 \!\!\shortrightarrow\!\! 1\!\!\shortrightarrow\!\! 1)+\ldots = \infty.$
\end{example}
%
\begin{figure}[]
    \centering
    \subfloat[][]{\begin{tikzpicture}[->, auto, node distance=0.8cm, every loop/.style={},
                    main node/.style={
                    }]
                     \node[main node] (0) [] {0};            
 \node[main node] (1) [right of=0] {1};
 \tikzset{every loop/.style={in=120,out=60,looseness=5}}
  \path[every node/.style={font=\sffamily\tiny}]
   (0) edge node [above]{1} (1)
  (1) edge [loop above] node [above] {1} (1);
\end{tikzpicture}\label{fig:RP}
}
    \subfloat[][]{%
        \centering
       \begin{tikzpicture}[->, auto, node distance=0.8cm, every loop/.style={},
                    main node/.style={
                    }]
                 
 \node[main node] (0) [] {$0$};
 \node[main node] (1) [right of=0] {$1$};
 \node[main node] (a) [right of=1] {$a$};
 \tikzset{every loop/.style={in=120,out=60,looseness=5}}
  \path[every node/.style={font=\sffamily\tiny}]
   (0) edge node [above]{1} (1)
  (1) edge node [above]{1/2} (a)
  (1) edge [loop above] node [above] {1/2} (1);
  
\end{tikzpicture}\centering
        \label{fig:a}
}
    ~ 
    \subfloat[][]{%
        \centering
        \begin{tikzpicture}[->, auto, node distance=0.8cm, every loop/.style={},
                    main node/.style={
                    }]
                 
 \node[main node] (a) [] {$a$};
 \node[main node] (0) [right of=a] {$0$};
 \node[main node] (1) [right of=0] {$1$};
 \node[main node] (b) [right of=1] {$b$};
 
  \path[every node/.style={font=\sffamily\tiny}]
  (0) edge node [above] {1/2} (a)
  (0) edge [bend left=45] node [above] {1/2} (1)
  
  (1) edge [bend left=45] node [below] {1/2} (0)
  (1) edge node [above] {1/2} (b);
\end{tikzpicture}
        \label{fig:b}
}
    ~ 
    \subfloat[][]{%
        \centering
        \begin{tikzpicture}[->, auto, node distance=1.1cm, every loop/.style={},
                    main node/.style={}]         
 \node[main node] (x1) [] {0};
 \node[main node] (x2) [right of=x1] {1};
 \node[main node] (x3) [right of=x2] {2};
 \node[main node] (x4) [right of=x3] {3};
  \node[main node] (x5) [right of=x4] {\ldots};
    \node[main node] (x6) [below right of=x4] {\ldots};
 \node[main node] (a) [below of=x1] {$a$};
  \path[every node/.style={font=\sffamily\tiny}]
  (x1) edge node [above] {$1\!-\!1/4$} (x2)
  (x1) edge node [left] {$1/4$} (a)
  (x2) edge node [above] {$1\!-\!1/4 ^2$} (x3)
  (x2) edge [out=270, in=0]node [left, near start] {$1/4^2$} (a)
  (x3) edge node [above] {$1\!-\!1/4^3$} (x4)  
    (x3) edge [out=270, in=0] node  [left, very near start]  {$1/4^3$} (a)
  (x4) edge node [above] {$1\!-\!1/4^4$} (x5) 
  (x4) edge [out=250, in=3] node  [left, very near start]  {$1/4^4$} (a)
;
\end{tikzpicture}
        \label{fig:c}
}
   \caption{PARS with different properties, see Table~\ref{tab:ex:examples-abcc}.}
    \label{fig:examples-abc}
\end{figure}
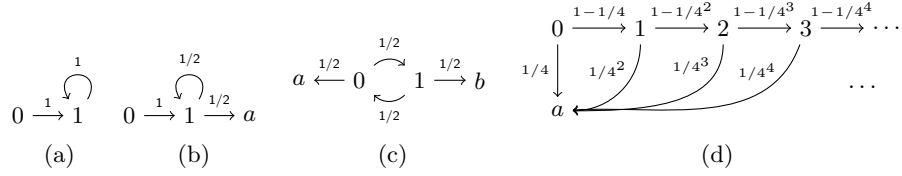
\noindent
The following proposition justifies that we refer to 
$P$ as a probability function. 
\begin{restatable}{proposition}{mainclaim}
\label{prop:prob-dist}
For an arbitrary finite path $\pi$, $1\ge P(\pi)>0$.
For every element $s$, $P(s \rightarrow^* \cdot)$ and $P(s\rightarrow^\infty)$ comprise a probability distribution,
i.e., $\forall t\in\NF{R}(s)\colon 0\leq P(s \rightarrow^* t)\leq 1$;
$0\leq P(s\rightarrow^\infty)\leq 1$; and
$\sum_{t\in\NF{R}(s)}P(s \rightarrow^* t) + P(s\rightarrow^\infty)=1$.
\end{restatable}
%
\begin{proof}The proofs are simple but lengthy and are given in the Appendix.
\end{proof}
%
%
Proposition \ref{prop:countable-infinity-summation} justifies that we refer to $P(s\rightarrow^\infty)$ as a probability of divergence.
\begin{restatable}{proposition}{secondclaim}\label{prop:countable-infinity-summation}
Consider a PARS which has an element $s$ for which $\Delta^\infty(s)$ is countable (finite or infinite).
Let $P(s_1\!\rightarrow\! s_2\!\rightarrow\! \cdots ) =
\prod_{i=1,2,\ldots}\!\! P(s_i\!\rightarrow\!s_{i+1})$ be the probability of an infinite path then $P(s\!\rightarrow^\infty)=\sum_{\delta\in\Delta^\infty(s)} P(\delta)$ holds.
%
%
\end{restatable}
\begin{proof}
See Appendix. 
\end{proof}
We can now
define \emph{probabilistic} and \emph{almost-surely} (abbreviated ``a-s.'') versions of important notions for derivation systems.
A system is
\begin{itemize}
\item \emph{almost-surely convergent} if for all $s_1 \leftarrow^{*} s \rightarrow^{*} s_2$ there is a normal form $t\in \NF{R}$ such that $s_1 \rightarrow^{*} t \leftarrow^{*} s_2$ and $P(s_1 \rightarrow ^{*} t)  = P(s_2 \rightarrow ^{*} t)= 1$,
\item \emph{locally almost-surely convergent} if for all $s_1 \leftarrow s \rightarrow s_2$ there is a $t\in \NF{R}$ such that $s_1 \rightarrow^{*} t \leftarrow^{*} s_2$ and $P(s_1 \rightarrow ^{*} t)  = P(s_2 \rightarrow ^{*} t)= 1$,
\item \emph{almost-surely terminating}\footnote{Almost-sure termination is named \emph{probabilistic termination} elsewhere, e.g.,~\cite{Sneyers2012,Fruhwirth2002}.} iff every element $s$ has $P(s\rightarrow^\infty) = 0$, and 
\item \emph{probabilistically normalizing} iff every element $s$ has a normal form $t$ such that
$P(s \rightarrow^* t) > 0$. 
\end{itemize}
We have deliberately omitted almost-sure confluence and local confluence~\cite{DBLP:conf/rta/BournezK02}, since these require a more advanced measure in order to define the probability of visiting a perhaps reducible element. 
%
\begin{table}
\begin{center}
\begin{tabular}{l@{\;\;}|@{\;\;}c@{\;}@{\;\;}c@{\;}@{\;\;}c@{\;}@{\;\;}c@{\;}@{\;\;}c@{\;}}
& {{{(a)}}} & {{{(b)}}}& {{{(c)}}}& {{{(d)}}}& {{{(d$'$)}}}\\
\hline
Loc.~confl. & +& +& +&+&+\\
Confl. & +& +& --& +&+\\
Term. & --& --& --& --&--\\
\hline
A-s.~loc.~conv. & --& + & --&--&+\\
A-s.~conv. & -- & + & --&--&+ \\
A-s.~term. & -- & +& +& --&+\\
\hline
\end{tabular}
\end{center}
\caption{A property overview of the systems \protect\subref{fig:RP}--\protect\subref{fig:c} in Figure \ref{fig:examples-abc} and (d$'$) with same ARS as (d), but with all probabilities replaced by $1/2$.}
\label{tab:ex:examples-abcc}
\end{table}
\begin{example}\label{ex:three-sample-pars}
The four probabilistic systems in Figure~\ref{fig:examples-abc} demonstrate these properties. 
\noindent We notice that \protect\subref{fig:a}--\protect\subref{fig:c} are normalizing in $\{a\}$, $\{a,b\}$ and $\{a\}$, respectively. Furthermore, they are all non-terminating: system \protect\subref{fig:a} and \protect\subref{fig:b} are {a-s.~terminating}, which is neither the case for \protect\subref{fig:RP} nor \protect\subref{fig:c}; for element $0$ in system \protect\subref{fig:c} we have $P(0\rightarrow^\infty) = \prod_{i=1}^{\infty} (1 - (1/4)^i) \approx 0.6885>0$.\footnote{Verified by Mathematica. The exact result is $\left(\frac{1}{4};\frac{1}{4}\right)_{\infty }$;
see~\cite{QPochhammer} for the definition of this notation.}
Table \ref{tab:ex:examples-abcc} summarizes their properties of (almost-sure) (local) confluence; 
(d$'$) refers to a  PARS with the same underlying ARS as \protect\subref{fig:c} and with all probabilities $=1/2$. 

System~\protect\subref{fig:b} is a probabilistic version of a classical example~\cite{DBLP:journals/jsyml/Hindley74,Huet1980} which demonstrates that termination (and not only a-s.~termination) is required in order for local confluence to imply confluence. The difference between system (d) and (d$'$) emphasizes that the choice of probabilities
do matter for whether or not different probabilistic properties hold. For any element
$s$ in (d$'$), the probability of reaching the normal form $a$ is $1/2+1/2^2+1/2^3+\cdots=1$.
\end{example}

%
%

\section{Properties of Probabilistic Abstract Reduction Systems}\label{sec:props-of-PARS}
With a focus on almost-sure convergence, we consider now relevant relationships between the properties of probabilistic and their underlying non-probabilistic systems.
%
%
Lemmas~\ref{lem:norm=probnorm} and \ref{lem:tem2probterm}, below, have previously been suggested by~\cite{DBLP:conf/rta/BournezK02} without proofs, and we have chosen to include them as well as their proofs to provide a better understanding of the nature of almost-sure convergence. The most important properties are
summarized as follows.
For any PARS $R^P$: 
\begin{itemize}
\item $R^P$ is {normalizing} if and only if  it is {probabilistically normalizing} (Lemma~\ref{lem:norm=probnorm}),
\item  if $R^P$ is almost-surely terminating then it is {normalizing} (Lemma~\ref{lem:probterm2norm}),
\item if $R^P$ is {terminating} then it is almost-surely terminating (Lemma~\ref{lem:tem2probterm}),
\item $R^P$ is almost-surely terminating and confluent, 
 if and only if it is almost-surely convergent (Theorem~\ref{thm:probterm+confl<=>probconfl}). 
\end{itemize}
\par\noindent
%
%
The following inductive characterization of the probabilities for reaching a given normal form is useful
for the proofs that follow.
\begin{proposition}\label{prop:probsucc*succprob=prob}
For any reducible element $s$,
the following holds.
$$\sum_{t \in \NF{R}} P(s \rightarrow^* t) = \sum_{s \rightarrow s'} \biggl(P(s \rightarrow s')\times\sum_{t \in \NF{R}}P(s' \rightarrow^* t)\biggr) $$
\end{proposition}
\begin{proof}
Any path from $s$ to a normal form $t$ will have the form $s\rightarrow s'\rightarrow\cdots\rightarrow t$, for some
direct successor $s'$ of $s$. The other way round, any normal form for a direct successor $s'$ of $s$ will also be
a normal form of $s$. With this observation, the proposition follows directly from Definition~\ref{def:pars}
(prob.\ of path).
\end{proof}

\begin{lemma}[{\cite{DBLP:conf/rta/BournezK02}}]\label{lem:norm=probnorm}
A PARS is {normalizing} if and only if it is {probabilistically normalizing}.
\end{lemma}
\begin{proof}
Every element $s$ in a normalizing PARS has a normal form $t$ such that $s \rightarrow^* t$ and by definition of PARS, $P(s \rightarrow^* t) > 0$, which makes it probabilistically normalizing. 
The other way round, the definition of probabilistic normalizing includes normalization.
\end{proof}
Prob. normalization differs from the other properties in nature (requiring probability $>0$ instead of $=1$), and is the only one which is equivalent to its non-probabilistic counterpart. Thus, the existing results on proving and disproving normalization can be used directly to determine probabilistic normalization. 
The following lemma is also a consequence of Proposition~7, parts 3 and 5, of {\cite{DBLP:conf/rta/BournezK02}}.
%
%

\noindent
\begin{lemma}\label{lem:probterm2norm}
If a PARS is almost-surely terminating then it is {normalizing}.
\end{lemma}
\begin{proof}
For every element $s$ in a almost-surely terminating system,
Proposition~\ref{prop:prob-dist} gives that $\sum_{t \in \NF{R}} P(s \rightarrow^* t)=1$, and hence $s$ has at least one normal form $t$ such that
$P(s \rightarrow^* t) > 0$. By Lemma \ref{lem:norm=probnorm}, the system is also normalizing. 
\end{proof}
The opposite is not the case, as demonstrated by system \protect\subref{fig:c} in Figure \ref{fig:examples-abc}; every element has a normal form, but the system is not almost-surely terminating.

%
%

\begin{lemma}[{\cite{DBLP:conf/rta/BournezK02}}]\label{lem:tem2probterm}
If a PARS is {terminating} then it is almost-surely terminating.
\end{lemma}
\begin{proof}
In a terminating PARS, $\Delta^\infty(s)=\emptyset$ for any element $s$.
By Proposition~\ref{prop:countable-infinity-summation} we have $P(s\rightarrow^\infty) =0$.
\end{proof}
The opposite is not the case, as demonstrated by systems~\subref{fig:a}--\subref{fig:c} in Figure~\ref{fig:examples-abc}.
The following theorem is a central tool for proving almost-sure convergence. 

\begin{theorem}\label{thm:probterm+confl<=>probconfl}
A PARS is almost-surely terminating and confluent 
 if and only if it is almost-surely convergent.
\end{theorem}
Thus, to prove almost-sure convergence of a given PARS, one may use the
methods of~\cite{Fioriti2015,Bournez05} to prove almost-sure termination and
prove classical confluence -- referring to Newman's lemma 
(cf.~our discussion in the Introduction), or using the method of mapping the system
into another system, already known to be confluent, as described in Section~\ref{sec:transformation}, below.

\begin{proof}[Theorem~\ref{thm:probterm+confl<=>probconfl}]
We split the proof into smaller parts, referring to properties that are shown below: 
``if'': by Prop. \ref{lem:probconfl2probterm} and Lemma \ref{lem:probconfl2confl}. ``only if'': by Lemma \ref{lem:probterm+conf2probconfl}.
\end{proof}

\begin{lemma}\label{lem:problocalconfl2probterm}
A PARS is almost-surely terminating if it is locally almost-surely convergent. 
\end{lemma}
\begin{proof}
Let $R^P$ be a PARS which is locally almost-surely convergent, and consider an arbitrary element $s$.
We must show $P(s\rightarrow^\infty)=0$ or, equivalently, $\sum_{t\in \NF{R}} P(s\rightarrow^* t)=1$.

When $s$ is a normal form, we have $P(s\rightarrow^* s) = 1$ and thus the desired property.
Assume, now, $s$ is not a normal form.
This means that $s$ has at least one direct successor; for any two (perhaps identical) direct successors $s', s''$,
local almost-sure convergence implies a unique normal form $t_{s', s''}$ of $s'$ as well as of $s''$
with $P(s' \rightarrow^* t_{s', s''}) = P(s'' \rightarrow^* t_{s', s''}) = 1$.
Obviously, this normal form is the same for all such successors and thus a unique normal form of $s$, so let us call it $t_s$.
We can now use Proposition~\ref{prop:probsucc*succprob=prob} as follows.
$$\sum_{t \in \NF{R}}\!\! P(s \shortrightarrow^* t) = 
P(s \shortrightarrow^* t_s) =
\sum_{\mathclap{s \rightarrow s'}} \biggl(\! P(s \shortrightarrow s')\cdot P(s' \shortrightarrow^* t_s)\!\biggr)
= \sum_{\mathclap{s \rightarrow s'}} P(s \shortrightarrow s') = 1.
$$
This finishes the proof.
%
%
%
\end{proof}
Since almost-sure convergence implies local almost-sure convergence, we obtain the weaker version of the above lemma.
\begin{proposition}\label{lem:probconfl2probterm}
A PARS is almost-surely terminating if it is almost-surely convergent.
\end{proposition}
%

\noindent The following property for (P)ARS, is used in the proof of Lemma~\ref{lem:probterm+conf2probconfl},
below.
\begin{proposition}
 \label{prelem:norm+conf2uniqNF} 
 A normalizing system is confluent if and only if every element has a unique normal form.
\end{proposition}
\begin{proof}
``If'': 
By contradiction: Let $R^P$ be a normalizing (P)ARS; assume that every element has a unique normal form and that $R^P$ is not confluent. By non-confluence, there exist $s_1 \leftarrow^* s \rightarrow^* s_2$ for which there does not exists a $t$ such that $s_1 \rightarrow^* t \leftarrow^* s_2$. However, $s$ has one unique normal form $t'$, i.e., $\{t'\}= \NF{R}(s)$. By definition of normal forms of $s$, we have that $\forall s'\colon s \rightarrow^* s' \Rightarrow \NF{R}(s) \supseteq \NF{R}(s')$. This holds specifically for $s_1$ and $s_2$, i.e.,  $\{t'\}=\NF{R}(s) \supseteq \NF{R}(s_1)$ and $\{t'\}= \NF{R}(s) \supseteq \NF{R}(s_2)$. Since $R$ is normalizing, every element has at least one normal form, i.e., $\NF{R}(s_1) \neq \emptyset$ and $\NF{R}(s_2) \neq \emptyset$, leaving one possibility: 
$ \NF{R}(s_1) = \NF{R}(s_2) = \{t'\}$. From this result we obtain
 $s \rightarrow^* s_1 \rightarrow^* t'$  and   $s \rightarrow^* s_2 \rightarrow^* t'$; contradiction. 
``Only if'': This is a known result; see, e.g.,~\cite{BaderNipkow1999}.
\end{proof}

\begin{lemma}\label{lem:probterm+conf2probconfl}
If a PARS is almost-surely terminating and confluent 
 then it is almost-surely convergent.
\end{lemma}
\begin{proof}
Lemma \ref{lem:probterm2norm} and Prop.\ \ref{prelem:norm+conf2uniqNF}  ensure that an a-s.\ terminating system has a unique normal form.
A-s.\ termination also ensures that this unique normal form is reached with probability $=1$, and thus the system is almost-surely convergent.
\end{proof}

\begin{lemma}\label{lem:probconfl2confl} 
A PARS is confluent if it is almost-surely convergent. 
\end{lemma}
\begin{proof}
Assume almost-sure convergence, then for each $s_1 \leftarrow^* s \rightarrow^* s_2$ there exists a $t$ (a normal form) such that $s_1 \rightarrow^* t \leftarrow^* s_2$.
\end{proof}

\section{Showing Probabilistic Confluence by Transformation}\label{sec:transformation}
 
The following proposition is a weaker formulation and consequence of Theorem \ref{thm:probterm+confl<=>probconfl}; it shows that 
(dis)proving confluence for almost-surely terminating systems is very relevant when (dis)proving almost-sure convergence.

\begin{proposition}\label{lem:probterm:probconfl=confl}
An almost-surely terminating PARS is almost-surely convergent if and only if it is confluent.
\end{proposition}
\begin{proof}
This is a direct consequence of Theorem \ref{thm:probterm+confl<=>probconfl} (or using Lemma \ref{lem:probterm+conf2probconfl} and \ref{lem:probconfl2confl}).
\end{proof}
%
Curien and Ghelli~\cite{CurienGhelli1991} presented a general method for proving confluence by 
 transforming\footnote{This is also referred to as interpreting a system elsewhere, e.g., \cite{CurienGhelli1991}.} the system of interest (under some restrictions) to a new system which is known to be confluent. We start by repeating their relevant result.
 \begin{lemma}[\cite{CurienGhelli1991}]\label{lem:CG}
Given two ARS $R=(A,\rightarrow_R)$ and $R'=(A,\rightarrow_{R'})$ and a mapping $G\colon A \rightarrow A'$, then $R$ is confluent if the following holds.
\begin{enumerate}[(C1)]
\item $R'$ is confluent,
\label{lem:CG:C1}
\item $R$ is normalizing, \label{lem:CG:C2}
\item if 
 $s \prescript{}{}{\rightarrow}_{{R}}\, t$ then 
$G(s) \prescript{}{}{\leftrightarrow}_{{R'}}^{*} 
 G(t)$\label{lem:CG:C3},
\item $\forall t \in \NF{R},\; G(t) \in \NF{R'}$, 
and\label{lem:CG:C4}
\item $\forall t,u \in \NF{R},\;G(t)=G(u)\,\Rightarrow\, t=u$\label{lem:CG:C5}
\end{enumerate}
\end{lemma}

\noindent
We present a version which permits also non-confluence of the transformed system to imply non-confluence of the original system. Notice that (C2)--(C5) is a part of (C2$'$)--(C5$'$), and in particular (C4$'$) requires additionally that only normal forms are mapped to normal forms. 
\begin{lemma}\label{lem:CG:conf=conf}
Given two ARS $R=(A,\rightarrow_R)$ and $R'=(A,\rightarrow_{R'})$ and a mapping $G\colon A \rightarrow A'$, satisfying
\begin{enumerate}[(C1$\,'$)]
\item (surjective) $\forall s' \in A', \;\;\exists s \in A, G(s) =s'$, 
\label{lem:CG:conf=conf:C1}
\item $R$ and $R'$ are normalizing,  \label{lem:CG:conf=conf:C2}
\item if $s \prescript{}{}{\rightarrow}_{{R}}\, t$ then 
$G(s) \prescript{}{}{\leftrightarrow}_{{R'}}^{*}\, G(t)$, and \\ 
if $G(s) \prescript{}{}{\leftrightarrow}_{{R'}}^{*}\, G(t)$ then $s \prescript{}{}{\leftrightarrow}_{{R}}^{*}\, t$,
\label{lem:CG:conf=conf:C3}
\item $\forall t \in \NF{R},\; G(t) \in \NF{R'}$, and 
$\forall t' \in \NF{R'},\; G^{-1}(t') \subseteq \NF{R}$,
\label{lem:CG:conf=conf:C4}
\item (injective on normal forms) ${\forall t,u \in \NF{R},\;G(t)=G(u)\Rightarrow t=u},$\label{lem:CG:conf=conf:C5}
\end{enumerate}
then $R$ is confluent iff $R'$ is confluent.
\end{lemma}
\begin{proof}
``$\Rightarrow$'': follows from Lemma \ref{lem:CG}.\\
``$\Leftarrow$'': Assume that $R$ is confluent and $R'$ is not confluent, i.e., there exist 
$s'_1 \prescript{}{}{\leftarrow}_{{R'}}^*\, s' \prescript{}{}{\rightarrow}_{{R'}}^*\, s'_2$ for which $\nexists t' \in R'\colon s'_1 \prescript{}{}{\rightarrow}_{{R'}}^*\,  t' \prescript{}{}{\leftarrow}_{{R'}}^*\,  s'_2$.\\\hbox to 1em{}By (C\ref{lem:CG:conf=conf:C2}$'$):
$\exists t'_1,t'_2 \in \NF{R'}\colon t'_1 \prescript{}{}{\leftarrow}_{{R'}}^*\, s'_1\prescript{}{}{\leftarrow}_{{R'}}^*\, s'  \prescript{}{}{\rightarrow}_{{R'}}^*\, s'_2 \prescript{}{}{\rightarrow}_{{R'}}^*\, t'_2 $ where $t'_1\neq t'_2$.\\
\hbox to 1em{}By (C\ref{lem:CG:conf=conf:C1}$'$) and (C\ref{lem:CG:conf=conf:C4}$'$): $\exists t_1,t_2 \in \NF{R} \colon G(t_1) = t'_1 \land G(t_2) = t'_2$\\
\hbox to 1em{}By (C\ref{lem:CG:conf=conf:C5}$'$): $t_1 \neq t_2$\\
\hbox to 1em{}By (C\ref{lem:CG:conf=conf:C3}$'$): $t'_1 \prescript{}{}{\leftrightarrow}_{{R'}}^*\, t'_2  
 \Rightarrow t_1 \prescript{}{}{\leftrightarrow}_{{R}}^*\, t_2  $\\
\hbox to 1em{}By confluence of $R$:  $t_1 = t_2$ (contradicts $t_1 \neq t_2$).
\end{proof}
We summarize the application of the above to probabilistic systems
in Theorems~\ref{thm:probterm+conflR'=>conflR} and~\ref{thm:tuttelut}.
\begin{theorem}\label{thm:probterm+conflR'=>conflR}
 An almost-surely terminating PARS $R^P=((A,\prescript{}{}{\rightarrow}_{{R}}),P)$ is {almost-surely convergent} 
if there exists an ARS $R'=(A',\prescript{}{}{\rightarrow}_{{R'}})$ and a mapping $G\colon A \rightarrow A'$ which together with $(A,\prescript{}{}{\rightarrow}_{{R}})$ satisfy (C\ref{lem:CG:C1})--(C\ref{lem:CG:C5}).
\end{theorem}
\begin{proof}
Since $R^P$ is a-s. terminating, $R$ is normalizing (Lemma \ref{lem:probterm2norm}). So, given an ARS $R'$ and $G$ be a mapping from $R$ to $R'$ satisfying (C\ref{lem:CG:C1}), (C\ref{lem:CG:C3})--(C\ref{lem:CG:C5}), we can apply Lemma \ref{lem:CG} and obtain that $R$ and thereby $R^P$ is confluent.
A-s.~convergence of $R^P$ follows from Prop.~\ref{lem:probterm:probconfl=confl}
since $R^P$ is confluent and a-s.\ terminating
\end{proof}

\begin{example}
We consider the nonterminating, {almost-surely terminating} system $R^P$ (below to the left) with the underlying normalizing system $R$ (below, middle), the confluent system $R'$ (below to the right) and the mapping $G(0) \!= 0$, $G({a})\! = {a}$.
%
%
%
$$
R^P\colon\;\begin{tikzpicture}[->, auto, node distance=1.2cm, every loop/.style={},
                    main node/.style={}, baseline=7pt]
\node[main node] (0) [] {0};
  \node[main node] (a) [right of=0] {{a}};  
 \path[every node/.style={font=\tiny}]
 (0) edge [loop above] node {p} (0)
  (0) edge node [above] {1-p} (a);
  \end{tikzpicture}
\qquad\quad
R\colon\;
\begin{tikzpicture}[->, auto, node distance=1.2cm, every loop/.style={},
                    main node/.style={}, baseline=7pt]
\node[main node] (0) [] {0};
  \node[main node] (a) [right of=0] {{a}};
  
 \path[every node/.style={font=\tiny}]
 (0) edge [loop above] node {} (0)
  (0) edge node [above] {} (a);
  \end{tikzpicture}
\qquad\quad
R'\colon\;
  \begin{tikzpicture}[->, auto, node distance=1.2cm, every loop/.style={},
                    main node/.style={}, baseline=-3pt]
\node[main node] (0) [] {0};
  \node[main node] (a) [right of=0] {{a}};
  
 \path[every node/.style={font=\tiny}]
  (0) edge node [above] {} (a);
  \end{tikzpicture}
$$
The systems $R$, $R'$ and the mapping $G$ satisfy (C\ref{lem:CG:C1})--(C\ref{lem:CG:C5}), and therefore we can conclude that $R^P$ is {almost-surely convergent}.
\end{example}

\begin{theorem}\label{thm:tuttelut}
Given an almost-surely terminating PARS $R^P=(R,P)$ with $R=(A,\prescript{}{}{\rightarrow}_{{R}})$, an ARS $R'=(A,\prescript{}{}{\rightarrow}_{{R'}})$ and a mapping $G$ from $A$ to $A'$ which together with $R$ satisfy (C\ref{lem:CG:conf=conf:C1}$'$)--(C\ref{lem:CG:conf=conf:C5}$'$), then system $R^P$ is {almost-surely convergent} if and only $R'$ is confluent.
%
\end{theorem}
\begin{proof} Assume notation as above.
Since $R^P$ is a-s.\ terminating, $R$ is normalizing (Lemma \ref{lem:probterm2norm}), thus satisfying the first part of (C\ref{lem:CG:C2}$'$). So, given an ARS $R'$ and $G$ be a mapping from $A$ to $A'$ which together with $R$ satisfy (C\ref{lem:CG:C1}$'$)--(C\ref{lem:CG:C5}$'$), we can apply Lemma \ref{lem:CG} obtaining that $R$ is confluent iff $R'$ is confluent.
Prop.\ \ref{lem:probterm:probconfl=confl} gives that the a-s.\ terminating $R^P$ is a-s.\ convergent iff  $R'$ is confluent.
\end{proof}

\section{Examples}\label{sec:examples}
In the following we show almost-sure convergence in two different cases that examplifies Theorem \ref{thm:tuttelut}. We use the existing method for showing  almost-sure termination~\cite{Fioriti2015,Bournez05}: 
To prove that a PARS $R^P= ((A,\rightarrow),P)$ is a-s.\ terminating, it suffices to show existence of a \emph{Lyapunov ranking function}, i.e., a measure  $\mathcal{V}: A \rightarrow \mathbb{R}^+$ where $\forall s \in A$ there exists an $\epsilon > 0$ so the \emph{inequality of $s$}, $\mathcal{V}(s) \geq \sum_{s \rightarrow s'} P(s \rightarrow s')\cdot\mathcal{V}(s') + \epsilon$ holds.


\subsection{A Simple, Antisymmetric Random Walk}
We consider $R^P = (R,P)$, depicted in Figure \ref{fig:examples-abc-gy-gy}\protect\subref{fig:rndwalk}, a simple positive antisymmetric 1-dimensional random walk.
In each step the value $n$ can either increase to $n+1$, $P(n\rightarrow n+1)=1/3$, or decrease to $n-1$  (or if at 0 we ``decrease'' to the normal form $a$ instead), $P(n \rightarrow n-1)=P(0 \rightarrow a) =2/3$. 
%
%
Formally, the underlying system $R=(A,\rightarrow)$ is defined by $A= \mathbb{N}\uplus\{a\}$ and $\rightarrow \; = \{0 \rightarrow a\} \uplus \{n \rightarrow n' \mid n,n' \in \mathbb{N}, n' = n +1 \lor n' = n -1\}$. 
%
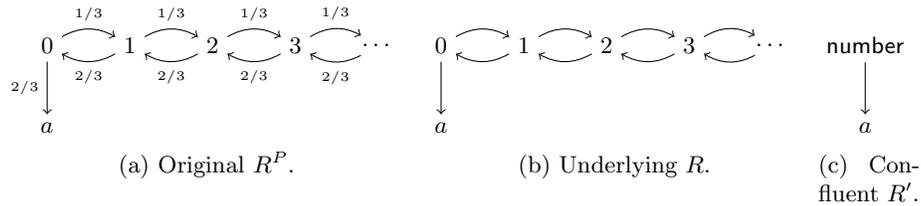
\begin{figure}[b]
    \centering
      \subfloat[][Original $R^P$.]{%
        \centering
        \begin{tikzpicture}[->, auto, node distance=1.1cm, every loop/.style={},
                    main node/.style={}]        
 \node[main node] (x1) [] {0};
\node[main node] (x2) [right of=x1] {1};
\node[main node] (x3) [right of=x2] {2};
\node[main node] (x4) [right of=x3] {3};
  \node[main node] (x5) [right of=x4] {\ldots};
\node[main node] (a) [below of=x1] {$a$};
  \path[every node/.style={font=\sffamily\tiny}]
  (x1) edge [bend left] node [above] {$1/3$} (x2)
  (x2) edge [bend left] node [below] {$2/3$} (x1)
  (x2) edge [bend left] node [above] {$1/3$} (x3)
  (x3) edge [bend left] node [below] {$2/3$} (x2)
  (x3) edge [bend left] node [above] {$1/3$} (x4)
  (x4) edge [bend left] node [below] {$2/3$} (x3) 
  (x4) edge [bend left] node [above] {$1/3$} (x5)
  (x5) edge [bend left] node [below] {$2/3$} (x4)     
  (x1) edge node [left] {$2/3$} (a)
;
\end{tikzpicture}
        \label{fig:rndwalk}
}
    ~
\subfloat[][Underlying $R$.]{%
        \centering
        \begin{tikzpicture}[->, auto, node distance=1.1cm, every loop/.style={},
                    main node/.style={}]        
 \node[main node] (x1) [] {0};
\node[main node] (x2) [right of=x1] {1};
\node[main node] (x3) [right of=x2] {2};
\node[main node] (x4) [right of=x3] {3};
  \node[main node] (x5) [right of=x4] {\ldots};
\node[main node] (a) [below of=x1] {$a$};
  \path[every node/.style={font=\sffamily\tiny}]
  (x1) edge [bend left] node [above] {} (x2)
  (x2) edge [bend left] node [below] {} (x1)
  (x2) edge [bend left] node [above] {} (x3)
  (x3) edge [bend left] node [below] {} (x2)
  (x3) edge [bend left] node [above] {} (x4)
  (x4) edge [bend left] node [below] {} (x3) 
  (x4) edge [bend left] node [above] {} (x5)
  (x5) edge [bend left] node [below] {} (x4)     
  (x1) edge node [left] {} (a)
;
\end{tikzpicture}
        \label{fig:rndwalk:R}
}
    ~
\subfloat[][Confluent $R'$.]{%
        \centering
        \begin{tikzpicture}[->, auto, node distance=1.1cm, every loop/.style={},
                    main node/.style={}]        
 \node[main node] (x1) [] {\textsf{number}};
\node[main node] (a) [below of=x1] {$a$};
  \path[every node/.style={font=\sffamily\tiny}]
  (x1) edge node [left] {} (a)
;
\end{tikzpicture}
        \label{fig:rndwalk:R'}
}
    \caption{Random Walk (1 Dimension)}
    \label{fig:examples-abc-gy-gy}
\end{figure}

\noindent
%
%
%
%
We start by showing $R^P$ a-s.\ terminating, i.e., that a Lyapunov ranking function exists: let the measure $\mathcal{V}$ be defined as follows.
$$\mathcal{V}(s) = \begin{cases}
s+2\text{,} & \text{if }s \in \mathbb{N} \\
1\text{,} & \text{if } s = a\\
 \end{cases}
$$
This
function
is a Lyapunov ranking
since the inequality (see above) holds for all elements $s \in A$; we divide into three cases $s>0$, $s=0$, and $s=a$:
\[\begin{array}{ll}
\mathcal{V}(s) > \frac{1}{3}\cdot \mathcal{V}(s+1) + \frac{2}{3}\cdot \mathcal{V}(s-1) &\Leftrightarrow
s+2 > \frac{1}{3}\cdot (s+3) + \frac{2}{3}\cdot (s+1) \;\;\; (= s+\frac{5}{3}) \\
\mathcal{V}(0) > \frac{1}{3}\cdot \mathcal{V}(1) + \frac{2}{3}\cdot \mathcal{V}(a) &\Leftrightarrow
2 > \frac{1}{3}\cdot 3 + \frac{2}{3}\cdot 1 \text{ , and } \\
\mathcal{V}(a) > 0 & \Leftrightarrow 1 >0.
\end{array}\]
%
Since $R^P$ is a-s.\ terminating, it suffice to define $R' = (\{\textsf{number},a\},\textsf{number} \rightarrow a)$, see Figure \ref{fig:examples-abc-gy-gy}\protect\subref{fig:rndwalk:R'}, and the mapping $G:\mathbb{N}\uplus\{a\} \rightarrow \{\textsf{number},a\}$.
$$G(s) = \begin{cases}
\textsf{number}\text{,} & \text{if }s \in \mathbb{N} \\
a\text{,} & \text{otherwise.}\\
 \end{cases}$$
Because $R^P$ is a-s.\ terminating, $R'$ is (trivially) a confluent system,
and the mapping $G$ satisfies (C\ref{lem:CG:conf=conf:C1}$'$)--(C\ref{lem:CG:conf=conf:C5}$'$) then $R^P$ is a-s.\ convergent (by Theorem~\ref{thm:tuttelut}).

\subsection{Herman's self-stabilizing Ring}
\emph{Herman's Ring}~\cite{Herman90} is an algorithm for self-stabilizing $n$ identical processors connected in an uni-directed ring, indexed 1 to $n$. Each process can hold one or zero tokens, and for each time-step, each  process either keeps its token or passes it to its left neighbour (-1) with probability $1/2$ of each event. When a process keeps its token and receives another, both tokens are eliminated. 

Herman showed that for an initial state with an odd number of tokens, the system will reach a stable state with one token with probability =1. This system is not almost-sure convergent, but proving it for a similar system can be a part of showing that Herman's Ring with 3 processes either will stabilize with 1 token with probability $=1$ or 0 tokens with probability $=1$. We use a boolean array to represent whether each process holds a token (\texttt{1} indicates a token) and is defined as in Figure \ref{fig:examples-hermans}\protect\subref{herman:org}, where both dashed and solid edges indicate reductions.

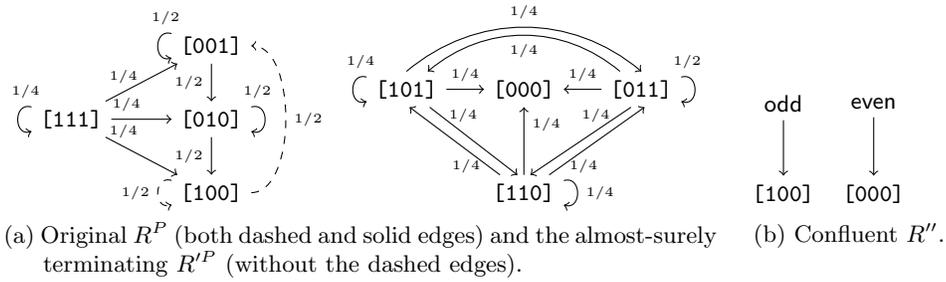
\begin{figure}
    \centering
      \subfloat[][Original $R^P$ (both dashed and solid edges) and the almost-surely \phantom{(a)} terminating $R'^P$ (without the dashed edges).]{%
        \centering
\begin{tikzpicture}[->, auto, node distance=1.2cm, every loop/.style={},  main node/.style={}]
  \node[main node] (111) [] {\texttt{[111]}};       
  \node[main node] (001) [above right=0.5cm and 0.8cm of 111] {\texttt{[001]}};
  \node[main node] (100) [below right=0.5cm and 0.8cm of 111]{\texttt{[100]}};
  \node[main node] (010) [right=0.8cm of 111] {\texttt{[010]}};
  \tikzset{every loop/.style={in=120,out=60,looseness=5}}
  \path[every node/.style={font=\sffamily\tiny}]
  (111) edge node [above,near start] {$1/4$} (001)
  (111) edge node [above,near start] {$1/4$} (010)
  (111) edge node [above,near start] {$1/4$} (100)
  (111) edge [loop left,in=197,out=163,looseness=2] node [above, very near start]{$1/4$} (111)
  (001) edge node [left] {$1/2$} (010)
 (001) edge [loop left,in=197,out=163,looseness=2] node [above, very near start]{$1/2$} (001)
%
  (010) edge node [left] {$1/2$} (100)
  (010) edge [loop left,in=343,out=17,looseness=2] node [above, very near start]{$1/2$} (010)
  ;
  \path[every node/.style={font=\sffamily\tiny}]
   (100) [dashed] edge [loop left,in=197,out=163,looseness=2] node [left]{$1/2$} (100)
  (100) [dashed] edge [bend right=120,in=270,out=270,looseness=0.8] node [right] {$1/2$} (001)
  ;
  \end{tikzpicture}
\begin{tikzpicture}[->, auto,node distance=1.0cm, every loop/.style={},  main node/.style={}]
  \node[main node] (000) [] {\texttt{[000]}};       
  \node[main node] (011) [right =0.5cm of 000] {\texttt{[011]}};
  \node[main node] (110) [below =0.9cm of 000]{\texttt{[110]}};
  \node[main node] (101) [left =0.5cm of 000] {\texttt{[101]}};
  \tikzset{every loop/.style={in=120,out=60,looseness=5}}
  \path[every node/.style={font=\sffamily\tiny}]
 
  (110) edge  node [left,very near start] {$1/4$} (101.280)
(101) edge node [right, very near start]{$1/4$} (110.110)
 
  (101.90) edge [bend left =40] node [above] {$1/4$} (011.90)
  (011) edge [bend right=40] node  {$1/4$} (101)
  
  (101) edge [loop left,in=197,out=163,looseness=2] node [above, very near start]{$1/4$} (101)
  (101) edge node [above] {$1/4$} (000)
  (011) edge node [above] {$1/4$} (000)
  (110) edge node [right, near end] {$1/4$} (000)
  (011) edge [loop left,in=343,out=17,looseness=2] node [above, very near start] {$1/2$} (011)

  (110) edge  node [right, very near start] {$1/4$} (011.280)
  (011.240) edge node [left,very near start] {$1/4$} (110.70)
 
 
  (110) edge [loop right,in=343,out=17,looseness=2] node [right] {$1/4$} (110)
%
  ;
  \end{tikzpicture}
  \label{herman:org}
  }
    \hspace{0.8em}
\subfloat[][\hbox to 4.8em{Confluent $R''$.\hss}]{%
        \centering
\begin{tikzpicture}[->, auto, node distance=1.2cm, every loop/.style={},  main node/.style={}]
  \node[main node] (odd) [] {\textsf{odd}}; 
  \node[main node] (100) [below of=odd] {\texttt{[100]}};
  \node[main node] (even) [right of=odd] {\textsf{even}}; 
  \node[main node] (000) [below of=even] {\texttt{[000]}};
  \path[every node/.style={font=\sffamily\tiny}]
  (odd) edge node [left] {} (100)
  (even) edge node [left] {} (000);
\end{tikzpicture}
\label{herman:new}
    }
      \caption{Herman's self-stabilizing Ring}
    \label{fig:examples-hermans}
\end{figure}
 
\noindent
Since \texttt{[000]} is a normal form  and $\{\mathtt{[100]},\mathtt{[010]},\mathtt{[001]}\}$ is the set of successor-states of each of \texttt{[100]},\texttt{[010]} and \texttt{[001]}, then 
we can prove stabilization of $R^P$ by showing almost-sure convergence for a slightly altered system $R'^P$, i.e., the system in Figure \ref{fig:examples-hermans}\protect\subref{herman:org} consisting of the solid edges only. 
 %

To show almost-sure convergence of $R'^P$, we prove almost-sure termination by showing the existence of a Lyapunov ranking function, namely $\mathcal{V}([b_1\,b_2\,b_3]) = 2^2\cdot (b_1+b_2+b_3) + b_1\cdot 2^0+b_2\cdot 2^{1}+b_3\cdot 2^2$, which decreases, firstly, with the reduction of tokens and, secondly, by position of the tokens. The only two states where $\mathcal{V}$ increases in a direct successor are \texttt{[110]} and \texttt{[101]} where the inequality of $\mathtt{[110]}$
 reduces to $11 > 9+\frac{1}{2}$ 
and that of $\mathtt{[101]}$
to $14 > 9+\frac{1}{2}$
showing $R^P$ to be a-s.~terminating.

We provide, now, a mapping $G$ from the elements of the underlying system into the elements of a trivially confluent system, i.e.,  $R''$ in Figure \ref{fig:examples-hermans}\protect\subref{herman:new}:
$$\begin{array}{rl}
G(\mathtt{[100]})= \mathtt{[100]} \qquad\hbox to 0.9em{} G(\mathtt{[000]}) &= \mathtt{[000]}\\
G(\mathtt{[111]}) = G(\mathtt{[001]}) = G(\mathtt{[010]}) &= \mathsf{odd}\\
G(\mathtt{[011]}) = G(\mathtt{[101]})= G(\mathtt{[110]})&=\mathsf{even}
\end{array}$$
The $R^P$ is a-s.\ term., $R''$ is confluent and $G$ satisfy (C\ref{lem:CG:conf=conf:C1}')--(C\ref{lem:CG:conf=conf:C5}'), then (by  Thm.~\ref{thm:tuttelut}) $R^P$ is {a-s.\ convergent}.

%
%

%
%
%
%
%
%
%
%
%
%
%
%
%
%
%
%

\section{Related Work}\label{sec:related}
We see our work as a succession of the earlier work by
Bournez and Kirchner~\cite{DBLP:conf/rta/BournezK02},
with explicit and simple definitions (instead of referring to Homogeneous Markov Chain theory) and proofs of central properties, and
showing novel properties that are important for showing (non-) convergence.
Our work borrows inspirations from the result of~\cite{Fruhwirth2002,Sneyers2010a,Sneyers2012}, given specifically
for probabilistic extensions of the programming languages CHR.
A notion of so-called nondeterministic PARS
have been introduced, e.g.,~\cite{Bournez05,Fioriti2015},
in which the choice of probability distribution for next reduction
is nondeterministic; these are not covered by our results.

PARS can be implemented directly in Sato's PRISM System~\cite{Sato1995,Sato2008},
which is a probabilistic version of Prolog,
and recent progress for
nonterminating programs~\cite{DBLP:journals/tplp/SatoM14} may be useful
convergence considerations.

\section{Conclusion}\label{sec:conclusion}
We have considered {almost-sure convergence} -- and how to prove it -- for probabilistic abstract reduction systems.
Our motivation is the application of such systems as computational systems having a deterministic input-output relationship,
and therefore {almost-sure termination} is of special importance.
We have provided properties that are useful when showing almost-sure (non-) convergence by consequence of other probabilistic and ``classic'' properties and by transformation.
We plan to generalize these results to {almost-sure convergence} modulo equivalence relevant for some Monte-Carlo Algorithms, that produces several correct answers (e.g. Simulated Annealing), and thereby continuing the work we have started for (non-probabilistic) CHR~\cite{Christiansen}.


\newpage
\appendix
\section{Selected proofs}

\mainclaim*

\begin{proof}Part one follows by Definition~\ref{def:pars}.
Part two is shown by defining a sequence of distributions $P^{(n)}, \; n\in \mathbb{N}$, only containing paths up to length $n$, and show that it converges to $P$.
Let
  $\Delta^{(n)}(s,t)$ be the subset of $\Delta(s,t)$ with paths of length $n$ or less, and
  $\Delta^{(n)}(s,\sharp)$ be the set of paths of length $n$, starting in $s$ and ending in a reducible element.
\par\noindent
We can now define $P^{(n)}$ over $\{\Delta^{(n)}(s,t)\mid t \in \NF{R}(s)\} \uplus \{\Delta^{(n)}(s,\sharp)\}$ as follows:
\begin{eqnarray}
\textstyle P^{(n)}(s \rightarrow^* t)  & = &  \textstyle\sum_{\delta \in \Delta^{(n)}(s,t)} P(\delta),\quad \mbox{and} \label{eqpn1}\\
P^{(n)}(s \rightarrow^\infty) & = & \textstyle \sum_{\pi\in\Delta^{(n)}(s,\sharp)} P(\pi).\label{eqpn2}
\end{eqnarray}
First, we prove by induction that $P^{(n)}$ is a distribution for all $n$.
The $P^{(0)}$ is a distribution because: 
\begin{inparaenum}[(i)]
\item If $s$ is irreducible, $P^{(0)}(s \rightarrow^* s)=1$ (the empty-path); and
$P^{(0)}(s \rightarrow^\infty)=0$ (a sum of zero elements).
\item  If $s$ is reducible, $P^{(0)}(s \rightarrow^* s)=0$; and
$P^{(0)}(s \rightarrow^\infty) = \sum_{s\rightarrow t} P(s\rightarrow t)=1$ by Definition \ref{def:pars}.
\end{inparaenum}
%
%
%
\par\noindent
The inductive step: The sets $\Delta^{(n+1)}(s,t)$,  $t\in\NF{R}(s)$, and $\Delta^{(n+1)}(s,\sharp)$ can be constructed by, for each path in $\Delta^{(n)}(s,\sharp)$, create its possible extensions by one reduction.
When an extension leads to a normal form $t$, it is added to $\Delta^{(n)}(s,t)$. Otherwise, i.e., if the new path leads to a reducible, it is included in $\Delta^{(n+1)}(s,\sharp)$. Formally, for any normal form $t$ of $s$, we write:
$$\begin{array}{lcl}
\Delta^{(n+1)}(s,t) & = &  \{ (s\!\shortrightarrow\!\cdots\shortrightarrow\! u\!\shortrightarrow\! t)\hbox to 0.42em{}\mid
               (s\!\shortrightarrow\!\cdots\shortrightarrow\! u)\in \Delta^{(n)}(s,\sharp),\,u\!\rightarrow\!t\} \;  \uplus \; \Delta^{(n)}(s,t) \\
\Delta^{(n+1)}(s,\sharp) &= & \{ (s\!\shortrightarrow\!\cdots\shortrightarrow\! u\!\shortrightarrow\! v)\mid
              (s\!\shortrightarrow\!\cdots\shortrightarrow\! u)\in \Delta^{(n)}(s,\sharp),\,u\!\shortrightarrow\!v,\,u\not\in\NF{R}(s)\} 
\end{array}$$
We show that for a given $s$, the probability mass added to the  $\Delta^{(\,\cdot\,)}(s,t)$ sets  is equal to the probability mass removed from
$\Delta^{(\,\cdot\,)}(s,\sharp)$ as follows (where $\delta_{s\!u} = (s\!\shortrightarrow\!\cdots\shortrightarrow\! u)$).
{\small\begin{eqnarray*}
&&\sum_{t\in\NF{R}(s)} P^{(\!n+1\!)}(s \!\rightarrow^{\!*}\! t) + P^{(\!n+1\!)}(s \rightarrow^{\!\infty})  = \sum_{\substack{t\in\NF{R}(s)\\ \delta \in \Delta^{(\!n+1\!)}(s,t)}} P^{(\!n+1\!)}(\delta)  + P^{(\!n+1\!)}(s \rightarrow^{\!\infty})  \\
&&=\quad\sum_{\mathclap{\substack{t\in\NF{R}(s)\\ \delta_{st} \in \Delta^{\!(\!n\!)}\!(s,t)}}} P^{(\!n\!)}(\delta)
\;\;+\;\; \sum_{\mathclap{\substack{\delta_{s\!u} \in \Delta^{\!(\!n\!)}\!(s,\sharp), \\ u\rightarrow v, v\in\NF{R}(s)}}} P^{(\!n\!)}(\delta)P(u\!\rightarrow\! v)
\;\;+\;\;  \sum_{\mathclap{\substack{\delta_{s\!u} \in \Delta^{\!(\!n\!)}\!(s,\sharp),\\  u \rightarrow v, v\not\in\NF{R}(s)}}} P^{(\!n\!)}\!(\delta)P(u\!\rightarrow\! v)\\
&&=\sum_{\mathclap{t\in\NF{R}(s)}} P^{(\!n\!)}\!(s \!\rightarrow^{\!*}\! t)
+ \sum_{\mathclap{\substack{\delta_{s\!u}\in \Delta^{\!(\!n\!)}\!(s,\sharp), \\ u\!\rightarrow\!v}}} P^{(\!n\!)}\!(\delta)P(u\!\rightarrow\! v) 
=\sum_{\mathclap{t\in\NF{R}(s)}} P^{(\!n\!)}\!(s \!\rightarrow^{\!*}\! t)
+ \sum_{\mathclap{\substack{\delta_{s\!u} \in \Delta^{\!(\!n\!)}\!(s,\sharp)}}} P^{(\!n\!)}(\delta)\!\!
\left(\!\sum_{\substack{ u\!\rightarrow\!v}}
P(u\!\rightarrow\! v)\!\right) \\
&&=\sum_{t\in\NF{R}(s)} P^{(\!n\!)}(s \!\rightarrow^{\!*}\! t)
+ P^{(\!n\!)}(s \rightarrow^{\!\infty}) = \mbox{\Large 1}
\end{eqnarray*}}
\par\noindent
Thus, for given $s$, $P^{(n+1)}$ defines a probability distribution. 
Notice also that the equations above indicate  
that
$P^{(n+1)}(s \rightarrow^* t) \geq P^{(n)}(s \rightarrow^* t)$, for all $t\in\NF{R}(s).$

Finally, for any $s$ and $t\in\NF{R}(s)$,
$\lim_{n\rightarrow\infty}\Delta^{(n)}(s,t) = \Delta(s,t),$
we get (as we consider increasing sequences of real numbers in a closed interval)\break
\hbox{$\lim_{n\rightarrow\infty}P^{(n)}(s\rightarrow^* t) = P(s\rightarrow^* t),$}
and as a consequence of this,\break
$
\lim_{n\rightarrow\infty}P^{(n)}(s\rightarrow^\infty) = P(s\rightarrow^\infty).
$ 
This finishes the proof.
\end{proof}

\secondclaim*

%
\begin{proof}
We assume the characterization in the proof of Proposition~\ref{prop:prob-dist} above,
of $P$ by the limits of the functions $P^{(n)}(s\rightarrow^* t)$
and $P^{(n)}(s\rightarrow^\infty)$ given by equations (\ref{eqpn1}) and (\ref{eqpn2}).
When $\Delta^\infty(s)$ is countable, $\lim_{n\rightarrow\infty}P^{(n)}(s\rightarrow^\infty)=\sum_{\delta\in\Delta^\infty(s)} P(\delta)$. 
\end{proof}
\end{document}